\pdfoutput=1

\documentclass{article}

\usepackage{geometry}
\usepackage{url}
\usepackage{algorithm}
\usepackage{algpseudocode}
\usepackage{amsmath,amssymb,amsfonts}
\usepackage[pdftex]{graphicx}
\usepackage{textcomp}
\usepackage{xcolor}
\usepackage{amsthm}
\usepackage{hyperref}
\hypersetup{pdfauthor={J. Nakamura et al.},pdftitle={Self-Stabilizing Construction of a Minimal Weakly ST-Reachable Directed Acyclic Graph}}

\usepackage{lscape}

\usepackage{authblk}

\newcommand{\argmin}{\mathop\textrm{arg~min}\limits}
\newcommand{\otherwise}{\textrm{otherwise}}

\newtheorem{theorem}{Theorem}
\newtheorem{lemma}{Lemma}

\theoremstyle{definition}

\newcommand{\calS}{\mathcal{S}}
\newcommand{\calT}{\mathcal{T}}
\newcommand{\calA}{\mathcal{A}}
\newcommand{\calB}{\mathcal{B}}
\newcommand{\calF}{\mathcal{F}}
\newcommand{\calP}{\mathcal{P}}

\newcommand{\var}[1]{\texttt{#1}}
\newcommand{\macro}[1]{\textit{#1}}
\newcommand{\true}{\textrm{true}}
\newcommand{\false}{\textrm{false}}
\newcommand{\red}{\textrm{red}}

\newcommand{\ra}{\rightarrow}
\newcommand{\vs}{\vspace{1mm}}
\newcommand{\tab}{\hspace{5mm}}

\begin{document}

\title{Self-Stabilizing Construction of a Minimal Weakly $\mathcal{ST}$-Reachable Directed Acyclic Graph\thanks{This work was supported by JSPS KAKENHI Grant Numbers 18K18000, 18K18029, 18K18031, and 20H04140; the Hibi Science Foundation; and Foundation of Public Interest of Tatematsu.\\ \indent\textcopyright\ 2020 IEEE. Personal use of this material is permitted. Permission from IEEE must be obtained for all other uses, in any current or future media, including reprinting/republishing this material for advertising or promotional purposes, creating new collective works, for resale or redistribution to servers or lists, or reuse of any copyrighted component of this work in other works.}}

\author[1]{Junya Nakamura\thanks{Corresponding author: junya[at]imc.tut.ac.jp}}
\author[2]{Masahiro Shibata}
\author[3]{Yuichi Sudo}
\author[4]{Yonghwan Kim}

\affil[1]{Toyohashi University of Technology, Japan}
\affil[2]{Kyushu Institute of Technology, Japan}
\affil[3]{Osaka University, Japan}
\affil[4]{Nagoya Institute of Technology, Japan}

\date{}

\maketitle

\begin{abstract}
We propose a self-stabilizing algorithm to construct a minimal weakly $\mathcal{ST}$-reachable directed acyclic graph (DAG), which is suited for routing messages on wireless networks.
Given an arbitrary, simple, connected, and undirected graph $G=(V, E)$ and two sets of nodes, senders $\mathcal{S} (\subset V)$ and targets $\mathcal{T} (\subset V)$, a directed subgraph $\vec{G}$ of $G$ is a weakly $\mathcal{ST}$-reachable DAG on $G$, if $\vec{G}$ is a DAG and every sender can reach at least one target, and every target is reachable from at least one sender in $\vec{G}$.
We say that a weakly $\mathcal{ST}$-reachable DAG $\vec{G}$ on $G$ is minimal if any proper subgraph of $\vec{G}$ is no longer a weakly $\mathcal{ST}$-reachable DAG. 
This DAG is a relaxed version of the original (or \emph{strongly}) $\mathcal{ST}$-reachable DAG, where every target is reachable from every sender.
This is because a strongly $\mathcal{ST}$-reachable DAG $G$ does not always exist; some graph has no strongly $\mathcal{ST}$-reachable DAG even in the case $|\mathcal{S}|=|\mathcal{T}|=2$.
On the other hand, the proposed algorithm always constructs a weakly $\mathcal{ST}$-reachable DAG for any $|\mathcal{S}|$ and $|\mathcal{T}|$.
Furthermore, the proposed algorithm is self-stabilizing; 
even if the constructed DAG deviates from the reachability requirement by a breakdown or exhausting the battery of a node having an arc in the DAG, this algorithm automatically reconstructs the DAG to satisfy the requirement again.
The convergence time of the algorithm is $O(D)$ asynchronous rounds, where $D$ is the diameter of a given graph.
We conduct small simulations to evaluate the performance of the proposed algorithm.
The simulation result indicates that its execution time decreases when the number of sender nodes or target nodes is large.
\end{abstract}

\section{Introduction}

Nowadays, wireless networks, e.g., Wireless Sensor Networks (WSN) and the Internet of Things (IoT), attract lots of attention in the area of distributed computing.
In a wireless network, generally, each node can communicate only with other nodes within its limited range.
Thus, routing a message from a sender node to a target (destination) node via intermediate nodes plays an important role.
In the literature, many routing algorithms for wireless networks have been proposed \cite{Green2003,Jacquet2001,MaheshKumar2016}.
In the routing task for wireless networks, the following properties are important due to the instability of nodes and their limited power source.
The first is the reachability between sender nodes and target nodes guaranteed by a routing algorithm. 
The second is the number of nodes necessary to participate in the task.
Realizing a reachability guarantee with fewer nodes is preferable because it can reduce the energy consumption of nodes.
The third is fault-tolerance.

For this routing task, Kim et al. proposed a construction algorithm of an \emph{$\calS\calT$-reachable directed acyclic graph (DAG)} \cite{Kim2019} on a wireless network with a set $\calS$ of sender nodes and a set $\calT$ of target nodes. 
This DAG provides reachability from every sender node $s \in \calS$ to every target node $t \in \calT$ with a minimal number of arcs.
However, they also proved in \cite{Kim2019} that constructing an $\calS\calT$-reachable DAG is not always possible.
A graph $G$ and sets $\calS$ and $\calT$ must satisfy a certain condition to have an $\calS \calT$-reachable DAG on $G$ even if we focus on the case $|\calS| \le 2$ and $|\calT| \le 2$.

In order to circumvent this impossibility, in this paper, we consider a weaker version of an $\calS\calT$-reachable DAG, called a \emph{weakly $\calS\calT$-reachable DAG}.
We say that in a directed graph $\vec{G}$, node $v$ is \emph{reachable} from $u$ or equivalently $u$ \emph{can reach} $v$ if there exists a directed path leading from $u$ to $v$ in $\vec{G}$.
A subgraph $\vec{G}$ of $G$ is a weakly $\calS\calT$-reachable DAG if (1) every sender node $s \in \calS$ can reach at least one target node $t \in \calT$, (2) every target node $t \in \calT$ is reachable from at least one sender node $s \in \calS$, and (3) $\vec{G}$ has no cycle.
Unlike an original (or \emph{strongly}) $\calS\calT$-reachable DAG, for any simple, connected, and undirected graph $G=(V,E)$ and two sets $\calS, \calT \subseteq V$, there \emph{always} exists a weakly $\calS\calT$-reachable DAG on $G$, as we prove later in this paper.

We propose a distributed algorithm that constructs a minimal weakly $\calS\calT$-reachable DAG, given a simple, connected, and undirected graph $G=(V, E)$ and two sets $\calS, \calT \subset V$.
The proposed algorithm guarantees the minimality of the constructed DAG $\vec{G}$, like \cite{Kim2019}.
In other words, if any arc is removed from the constructed digraph $\vec{G}$, the resulting digraph is no longer a weakly $\calS\calT$-reachable DAG.
Also, the algorithm is self-stabilizing \cite{Dolev2000}; it tolerates any number of transient failures of nodes.
More specifically, from arbitrary initial configurations, this algorithm eventually reaches a legitimate configuration in which the requirement of minimal weakly $\calS\calT$-reachable DAG is satisfied.
Therefore, even if the constructed DAG deviates from the reachability requirement by a breakdown or exhausting the battery of a node having an arc in the DAG, this algorithm automatically reconstructs a minimal weakly $\calS\calT$-reachable DAG. Then, the reachability is guaranteed again.
The convergence time of the proposed algorithm is $O(D)$ (asynchronous) rounds, and each node requires $O(\log D + \Delta)$ bits memory, where $D$ and $\Delta$ are the diameter and the maximum degree of a given graph $G$, respectively.

To summarize, the contribution of this paper is as follows:
\begin{itemize}
    \item defining a minimal weakly $\calS\calT$-reachable DAG, which is suitable for the routing messages in wireless sensor networks,
    \item proposing a self-stabilizing algorithm that constructs a minimal weakly $\calS\calT$-reachable DAG for any numbers of sender nodes and target nodes, 
    \item proving the correctness and the theoretical performance of the proposed algorithm, and
    \item evaluating and analyzing the performance of the algorithm by simulation.
\end{itemize}

This paper is organized as follows:
Section \ref{sec:related-work} discusses related work that constructs some kinds of DAG between sender nodes and target nodes.
Section \ref{sec:preliminaries} defines our computation model and the construction problem.
Section \ref{sec:proposed-algorithm} proposes a self-stabilizing algorithm to construct an $\calS\calT$-reachable DAG.
Section \ref{sec:correctness} proves the correctness and the theoretical performance of the proposed algorithm.
Section \ref{sec:evaluation} evaluates the performance of the proposed algorithm by simulation.
Finally, Section \ref{sec:conclusion} concludes this paper.

\section{Related Work}
\label{sec:related-work}

\begin{landscape}
\begin{table*}[t!]
	\centering
	\caption{
		Summary of the related DAG construction distributed algorithms.
	}
	\label{tab:related-work}

	\small
	\begin{tabular}{ccccccccc}
		\hline
		& Reachability & $|\calS|$&$|\calT|$& Topology & Arc Assignment$^{*1}$ & Fault-Tolerance & Anonymity & Time Complexity$^{*2}$$^{*3}$ \\
\hline

\cite{Aranha1996}		& Strong	& 1		& 1		& Biconnected	& All		& N/A	& Identified	& $O(n)$	\\
\cite{Chaudhuri2008}	& Strong	& 1		& 1		& Biconnected	& All		& Self-Stabilizing	& Identified	& $O(n \log n)$ \\
\cite{Karaata2002}		& Strong	& 1		& 1		& Biconnected	& All		& Self-Stabilizing  & Identified	& $O(D)$ \\
\cite{Kim2017,Kim2018b}	& Strong	& 1		& 1		& Connected		& Maximal	& Self-Stabilizing	& Semi-anonymous$^{*4}$	& $O(D)$ \\
\cite{Kim2018a}			& Strong	& 1		& 2		& Connected		& Maximal	& Self-Stabilizing	& Anonymous	& $O(D)$ \\
\cite{Kim2018a}			& Weak	& 2		& 2		& Connected		& Maximal	& Self-Stabilizing	& Anonymous	& $O(D)$ \\
\cite{Kim2020}			& Weak	& Any	& Any	& Connected$^{*5}$		& Maximal	& Self-Stabilizing	& Anonymous	& $O(\max(D, \Delta (|\calS| + |\calT|)))$ \\
\cite{Kim2019}			& Strong	& $\le 2$		&  $\le 2$		& Connected$^{*6}$	& Minimal	& Self-Stabilizing	& Semi-anonymous$^{*4}$	& $O(n)$ \\
Our result              & Weak	& Any	& Any	& Connected		& Minimal	& Self-Stabilizing	& Anonymous	& $O(D)$ \\
\hline
	\end{tabular}
	\begin{center}
		\begin{minipage}{180mm}
			\footnotesize
			${}^{*1}$: This column indicates how many edges on a given network are changed to arcs by each algorithm.

			${}^{*2}$: The time complexities are measured in terms of synchronous or asynchronous rounds.%

			${}^{*3}$: In each cell, $n$, $D$, and $\Delta$ are the total number of nodes, the diameter of a given graph, and the maximum degree of the graph.

			${}^{*4}$: All the sender and target nodes have global, unique identifiers, while the other nodes are anonymous.

			${}^{*5}$: The algorithm \cite{Kim2020} has a necessary condition to construct a weakly $\calS\calT$-reachable DAG.

			${}^{*6}$: Each node detects an error if a graph does not satisfy a necessary condition to construct an $\calS\calT$-reachable DAG.
		\end{minipage}
	\end{center}
\end{table*}
\end{landscape}

Table \ref{tab:related-work} summarizes the related algorithms that construct some kinds of DAG from sender nodes to target nodes on a given graph.
The most important aspect of Table \ref{tab:related-work} is reachability.
A DAG with strong reachability ensures that every target node is reachable from every sender node.
On the other hand, a DAG with weak reachability guarantees that every sender can reach at least one target node, and every target node is reachable from at least one sender node; thus, a sender node may not be able to reach some target node.

The first three algorithms \cite{Aranha1996,Chaudhuri2008,Karaata2002} in Table \ref{tab:related-work} construct a \emph{transport net} \cite{Karaata2002}, which ensures strong reachability for a sender node and a target node, on a given biconnected graph.
The construction of a transport net by the algorithms is based on the technique called \emph{st-ordering} \cite{Even1976} (also known as st-numbering), which assigns a number (called st-order or st-number) to each node, and a transport net can easily be constructed from these numbers.
The latter two algorithms \cite{Chaudhuri2008,Karaata2002} are self-stabilizing \cite{Dolev2000}; thus, the algorithms tolerate any number of transient faults.

Kim et al.~considered construction of another type of DAG called \emph{$(\sigma, \tau)$-directed acyclic mixed graph (DAMG)} \cite{Kim2017,Kim2018b,Kim2018a,Kim2020} where $\sigma$ and $\tau$ are the numbers of sender nodes and target nodes, respectively\footnote{Therefore, $\sigma = |\calS|$ and $\tau = |\calT|$.}.
The reachability provided by a DAMG depends on $\sigma$ and $\tau$: strong reachability for $\sigma = 1$ and $\tau = 1$ or $2$, and weak reachability for any $\sigma$ and any $\tau$.
These algorithms cannot assign a direction to every edge since the algorithms construct a DAMG on a given connected graph, unlike a transport net.
Thus, the algorithms ensure that the maximal edges of a constructed DAMG have directions.

Our previous work \cite{Kim2019} introduced a new graph structure called an $\calS\calT$-reachable DAG that ensures strong reachability and presented a self-stabilizing construction algorithm for this graph.
The algorithm focuses on decreasing the number of directed edges (i.e., arcs) of a constructed graph to provide strong reachability as possible as it can, different from the previous algorithms.
Satisfying the reachability requirement with the minimal number of arcs is important for wireless networks.
The reason is as follows.
A node that has an incoming or outgoing arc must always be active for routing messages.
However, such an active node consumes large amounts of power while its energy capacity is limited.
The minimality guarantee of the arcs can reduce the energy consumption of nodes and makes the lifetime of the routing function on wireless networks longer by rerouting the message delivery path that contains an exhausted node automatically with the self-stabilizing algorithm.

In this paper, we propose a self-stabilizing construction algorithm \textsf{MWSTDAG} for a minimal weakly $\calS\calT$-reachable DAG that ensures weak reachability.
By weakening the reachability requirement from an $\calS\calT$-reachable DAG, the algorithm can construct such a DAG for any numbers of sender nodes and target nodes.
The time complexity of the proposed algorithm is $O(D)$, where $D$ is the diameter of a given graph, which is faster than \cite{Kim2020}.

\section{Preliminaries}
\label{sec:preliminaries}

\subsection{Computation Model}
Let $G = (V, E, \calS, \calT)$ be a simple and connected graph where $V$ is the set of $n$ computational entities called \emph{nodes} (or \emph{processes}), and $E$ is the set of $m$ undirected \emph{edges} between nodes.
The graph is anonymous; that is, we do not assume the existence of globally unique identifiers for any nodes.
We call nodes in $\calS$ \emph{senders} and ones in $\calT$ \emph{targets}.
Furthermore, we denote $\calS = \{s_1, s_2, \dots, s_x\}$ and $\calT = \{t_1, t_2, \dots, t_y\}$.
Here, we assume $|\calS|, |\calT| \geq 1$ and $\calS \cap \calT = \emptyset$, i.e., $\calS$ and $\calT$ do not contain the same node\footnote{We made this assumption only for simplicity.
If there exists some node $v \in \calS \cap \calT$, then we can deal with $v$ as if $v$ is a normal node, that is, $v \in V \setminus \{\calS \cup \calT\}$ because $v$ is reachable from at least one sender ($v$ itself) and at least one target ($v$ itself) is reachable from $v$.}.
For any node $v$, we denote the set of $v$'s neighbors by $N(v)$, i.e., $N(v) = \{ u \in V \mid \{u, v\} \in E\}$.
Each node $v$ can distinguish its neighbors with unique local labels $l_v: N(v) \rightarrow \{1,2,\dots,\delta_v\}$, where $\delta_v = |N(v)|$ is the degree of node $v$.
Labels for node $v$ are independent of those of $v$'s neighbors, i.e., we assume nothing on the relation between $l_u(v)$ and $l_v(u)$.
However, we assume that $v$ knows $l_u(v)$ for every neighbor $u \in N(v)$.
For simplicity, we denote $l_v(u)$ by just $u$ in the pseudocode of node $v$.
The graph $G = (V, E, \calS, \calT)$ may be denoted by $G = (V, E)$ if senders and targets are not referred.

In this paper, we employ the state reading model of computation, where each node reads its own variables and those of its neighbors and updates only its own variables in an atomic action.
This model is commonly used with self-stabilizing algorithms.
An algorithm $\calA$ is defined by a set of variables each node has and a set of atomic actions defining how the variables are updated based on the values of its own variables and those of its neighbors.
An atomic action is denoted by the following form: $\langle label \rangle :: \langle guard  \rangle \rightarrow \langle statement \rangle$.
The label is used to identify each action.
The guard is a boolean predicate to specify when the following statement can be executed.
The statement is a sequence of assignments on variables of a node.
We say that an action is \emph{enabled} if the guard of the action is true.
We also say that a node is \emph{enabled} if the node has at least one enabled action; otherwise, the node is said to be \emph{disabled}.

The \emph{state} of a node consists of the values of all variables in the node.
A \emph{configuration} is a vector of states of all nodes.
Let $V'$ be a non-empty subset of $V$ and $\calA$ be an algorithm.
We denote $C \mapsto_{(V', \calA)} C'$ if a configuration $C'$ is obtained when each node in $V'$ performs an atomic action of $\calA$ in configuration $C$.
A \emph{schedule} is an infinite sequence $V_0, V_1, \dots$ of non-empty subsets of $V$.
An execution $\Xi_\calA (S, C_0)$ of algorithm $\calA$ along schedule $S = V_0, V_1, \dots$ starting from a configuration $C_0$ is uniquely defined as the infinite sequence $C_0, C_1, \dots$ of configurations such that $C_i \mapsto_{(V_i, \calA)} C_{i+1}$ for all $i \geq 0$.
We say that a schedule is \emph{(weakly) fair} if each node in $V$ appears infinitely often in the schedule.
We call an execution along a fair schedule a \emph{fair execution}.

\subsection{Self-Stabilization and Silence}
Algorithm $\calA$ is said to be \emph{self-stabilizing} for a problem $\mathcal{P}$ if there exists a set $\calF$ of configurations that satisfies the following three conditions:
\begin{itemize}
	\item \textbf{Convergence:}
		every fair execution of $\calA$ starting from any configuration eventually reaches a configuration in $\calF$.
	\item \textbf{Closure:}
		a configuration in $\calF$ never changes to a configuration out of $\calF$ according to $\calA$, i.e., there do not exist $C \in \calF$, $C' \notin \calF$, and $V' \subseteq V$ such that $C \mapsto_{(V', \calA)} C'$.
\item \textbf{Correctness:}
		each configuration in $\calF$ satisfies the specification of problem $\calP$.
\end{itemize}
A configuration in $\calF$ is called \emph{legitimate}, and a configuration not in $\calF$ is called \emph{illegitimate}.
A configuration $C$ is called \emph{final} if $C \mapsto_{(V', \calA)} C$ holds for any non-empty subset $V'$, and an algorithm $\calA$ is called \emph{silent} if every fair execution reaches a final configuration.

\subsection{Problem Specification}
\label{sec:problem-spec}

Let $G' = (V', E')$ be any (undirected) subgraph of $G$ (note that $G = G'$ may hold).
A directed graph or digraph $\vec{G} = (V'', A)$ is called a directed subgraph of $G'$ if a set $V''$ of nodes satisfies $V'' \subseteq V'$ and a set $A$ of arcs satisfies $(u, v) \in A \Rightarrow \{u, v\} \in E'$ where $(u, v)$ denotes an arc from $u$ to $v$.
A directed subgraph $\vec{G}$ of $G = (V, E, \calS, \calT)$ is called a \emph{weakly $\calS\calT$-reachable DAG} of $G$ if the following conditions hold:
\begin{itemize}
	\item \textbf{C1:} every sender in $\calS$ can reach at least one target in $\calT$ in $\vec{G}$;
	\item \textbf{C2:} every target in $\calT$ is reachable from a sender in $\calS$ in $\vec{G}$; and
	\item \textbf{C3:} there is no directed cycle in $\vec{G}$.
\end{itemize}
Here, for any two nodes $u$ and $v$, we say that $v$ is \emph{reachable} from $u$, or $u$ \emph{can reach} $v$ in $\vec{G}$ if there exists a directed path from $u$ to $v$ in $\vec{G}$.
Also, a weakly $\calS\calT$-reachable DAG $\vec{G}$ of $G$ is called \emph{minimal} if condition C1 or C2 becomes unsatisfied if any arc of $\vec{G}$ is removed.

Each node $v$ is assumed to have \emph{output variables} $v.\var{arc}[u] \in \{\false, \true\}$ for each neighbor $u \in N(v)$.
For each $\{u, v\} \in E$, $v.\var{arc}[u] = \true$ means that arc $(v,u)$ exists in the output digraph.
Specifically, for any configuration $C$, $\vec{G}(C) = (V, A(C))$ is defined as the digraph where the set of nodes is $V$, and the set of arcs is $\{(v,u) \mid \{u,v\} \in E \wedge v.\var{arc}[u]\}$ in configuration $C$.
A configuration $C$ is said to satisfy the specification of the minimal weakly $\calS\calT$-reachable DAG construction if digraph $\vec{G}(C)$ is a minimal weakly $\calS\calT$-reachable DAG of $G = (V, E, \calS, \calT)$.

\subsection{Time Complexity}
We measure \emph{time complexity} as the number of (asynchronous) \emph{rounds} of an execution.
Let $C_0$ be any configuration of $\calA$ and $S = V_0, V_1, \dots$ be any fair schedule.
The \emph{first round} of execution $\Xi = \Xi_\calA (S, C_0) = C_0, C_1, \dots$ is defined as the smallest prefix, say $C_0, C_1, \dots, C_t$ of $\Xi$ such that every enabled node in $C_0$ executes at least one action or becomes disabled by state changes of its neighbor nodes in the first $t$ steps, i.e., $\forall v \in Enabled(C_0), \exists i < t, v \in V_i \vee (v \in Enabled(C_{i-1}) \wedge v \notin Enabled(C_i))$, where $Enabled(C)$ is a set of all the enabled nodes in a configuration $C$.
The second round of $\Xi$ is defined as the first round of $\Xi'$, where $\Xi'$ is the suffix of $\Xi$ starting from $C_t$, that is, $\Xi' = C_t, C_{t+1}, \dots$, and so on.

\subsection{Hierarchical Collateral Composition}
\emph{Hierarchical collateral composition} \cite{Datta2013,Altisen2019} is used to combine self-stabilizing algorithms to build the proposed algorithm.
This composition is a variant of the collateral composition \cite{Herman1992} and can be defined as follows:
Let $\calA$ and $\calB$ be two distributed algorithms.
The hierarchical collateral composition of $\calA$ and $\calB$ is the distributed algorithm $\calB \circ \calA$, where the local algorithm of every node $p$, noted $(\calB \circ \calA)(p)$, is defined as follows:
\begin{itemize}
	\item $(\calB \circ \calA)(p)$ has all variables of $\calA(p)$ and $\calB(p)$.
	\item $(\calB \circ \calA)(p)$ has all actions of $\calA(p)$.
	\item Every action $L_i :: G_i \rightarrow S_i$ of $\calB(p)$ is rewritten in $(\calB \circ \calA)(p)$ as the action $L_i :: \neg C_p \wedge G_i \rightarrow S_i$ of $\calB(p)$, where $C_p$ is the disjunction of all guards of all actions in $\calA(p)$.
\end{itemize}
Roughly speaking, the hierarchical collateral composition assigns explicit priorities to actions of the original distributed algorithms $\calA$ and $\calB$, that is, any actions of the high layer algorithm $\calB$ are not allowed to be enabled until every action of the low layer algorithm $\calA$ becomes disabled.
Such priorities allow us to avoid problems caused by the nondeterminism of enabled actions and to achieve an efficient composite algorithm, as demonstrated in \cite{Devismes2016}.

\section{Proposed Algorithm}
\label{sec:proposed-algorithm}

In this section, we propose a self-stabilizing algorithm called \textsf{MWSTDAG} that constructs a minimal weakly $\calS\calT$-reachable DAG on a given graph $G = (V, E, \calS, \calT)$.
This algorithm is built by the hierarchical collateral composition \cite{Datta2013} and has four layers.
The first layer algorithm \textsf{L1SpanningForest} builds Breadth-First-Search (BFS) trees on a given network $G$ and checks reachability from sender nodes to target nodes on the trees.
The second layer algorithm \textsf{L2SpanningForest} builds another kind of BFS trees to ensure reachability to target nodes from the sender nodes that cannot reach any target node in the layer 1 trees.
The third layer algorithm \textsf{L3WSTDAG} constructs a (possibly non-minimal) weakly $\calS\calT$-reachable DAG based on the trees constructed in the first and second layers.
The final layer algorithm \textsf{L4RedundantArcRemoval} detects and removes redundant arcs in the DAG to guarantee the minimality of the generated weakly $\calS\calT$-reachable DAG.
Thus, $\textsf{MWSTDAG} = \textsf{L4RedundantArcRemoval} \circ \textsf{L3WSTDAG} \circ \textsf{L2SpanningForest} \circ \textsf{L1SpanningForest}$.
Hereafter, we call a BFS tree constructed in layer 1 (resp.~layer 2), an \emph{L1 tree} (resp.~an \emph{L2 tree}).

In the proposed algorithm, each node may have \emph{red} or \emph{blue} color.
The red color assigned in layer 1 indicates that the node can reach a target node by tracing an L1 tree.
The blue color assigned in layer 2 indicates that the node can reach a red node.
Thus, a blue node can reach a target node through the red node (if the configuration is legitimate for layers 1 and 2).
These colors are propagated from a lower node to a higher node in their L1 and L2 trees.
Note that these color assignments are conducted only in layers 1 and 2, and the other layers never change node colors.

The removal of redundant arcs in layer 4 plays an important role in guaranteeing the minimality of a constructed weakly $\calS\calT$-reachable DAG and is conducted with the following idea:
If a red node $v$ has at least one incoming arc from a blue node, we can remove all but one arc from red nodes to $v$ without violating the reachability requirement of a weakly $\calS\calT$-reachable DAG.
However, there is an exception.
If a sender node becomes unreachable to any target node because of the removal, the algorithm must not remove such an arc.
Because only a red node $u$ having two or more outgoing arcs to red nodes can detect such an arc, we propagate information about whether there is a red node having an incoming arc from a blue node from $u$'s descendant nodes to $u$ in an L1 tree.
Note that the discussion above is only for a red node; none of the blue nodes have any redundant arc in the algorithm.

Algorithms \ref{alg:l1algorithm}--\ref{alg:l4algorithm} are the pseudocodes of each layer algorithm.
To avoid two or more actions of the same layer becoming enabled, priorities to actions are assigned as follows: an action that appears earlier in each pseudocode has higher priority.
If guards of two or more actions are true, only the action with the highest priority among them becomes enabled.
Moreover, since we assume the hierarchical collateral composition, actions of a layer never become enabled until all actions in every lower layer are disabled.
Therefore, at most one action is enabled in a node during an execution of the proposed algorithm.

A node $v$ has two variables that represent two DAGs: $v.\var{l3\_arc}$ for a (possibly non-minimal) weakly $\calS\calT$-reachable DAG in layer 3, and $v.\var{arc}$ for a minimal weakly $\calS\calT$-reachable DAG as the output of the algorithm.
A node $v$ also has two variables representing its color: $v.\var{l1\_color}$ for red and $v.\var{l2\_color}$ for blue.

\begin{figure}[tp]
	\centering
	\includegraphics[scale=0.425]{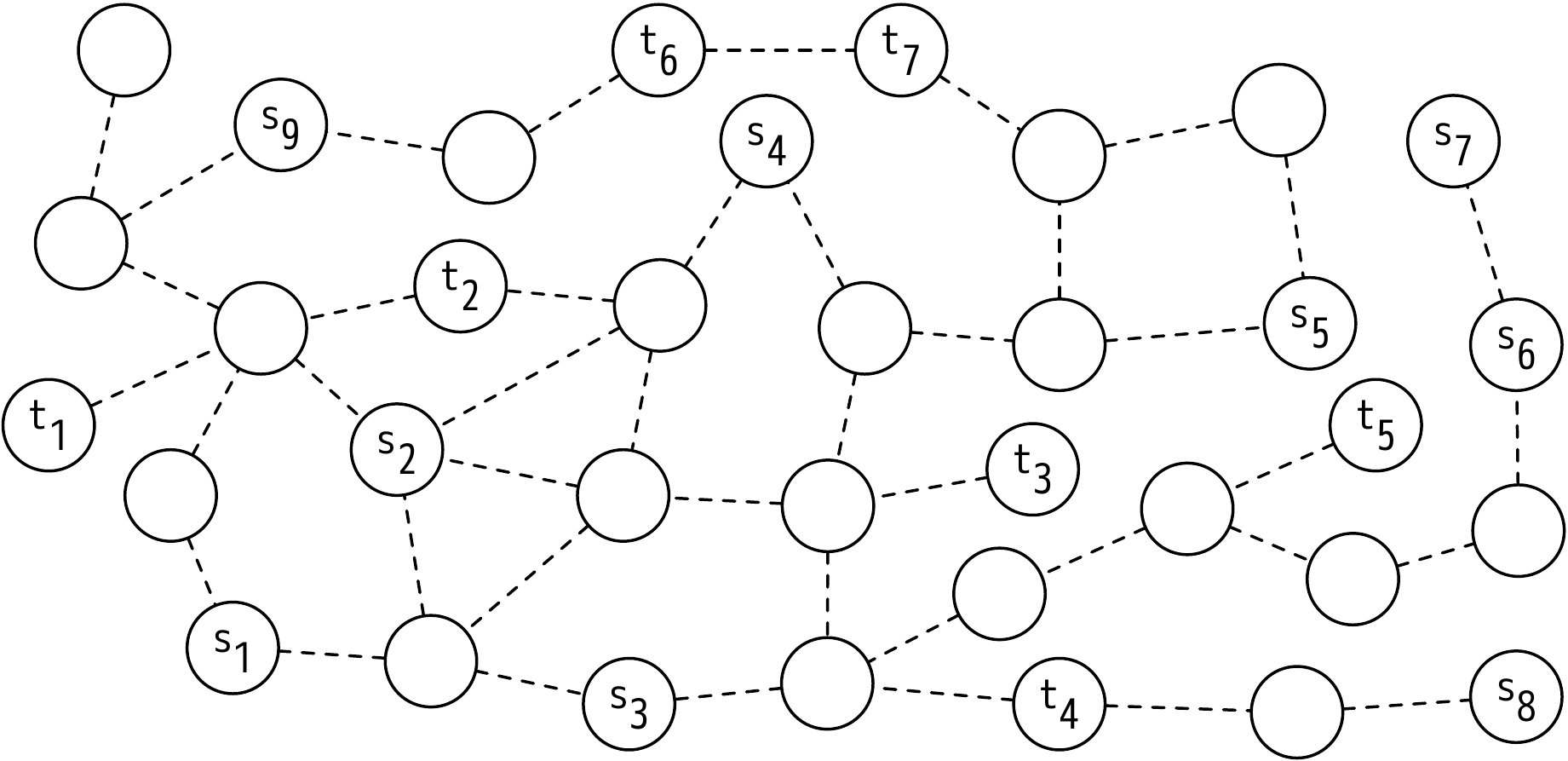}
	\caption{A given example graph $G = (V, E, \calS, \calT)$ where $\calS = \{ s_1, s_2, \dots, s_9 \}$ and $\calT = \{t_1, t_2, \dots, t_7 \}$.
		A dashed line represents an edge $e \in E$.}
	\label{fig:example-original}
\end{figure}
The followings are brief explanations of each layer.
Figures \ref{fig:example-original}--\ref{fig:example-finish} illustrate an example execution of the algorithm:
Fig.~\ref{fig:example-original} is a given undirected graph $G$, Figs.~\ref{fig:example-layer1}--\ref{fig:example-layer4} are the legitimate configurations for each layer, and Fig.~\ref{fig:example-finish} is the constructed minimal weakly $\calS\calT$-reachable DAG $\vec{G}$.

\textbf{Layer 1 (Fig.~\ref{fig:example-layer1}):}
\begin{figure}[tp]
	\centering
	\includegraphics[scale=0.425]{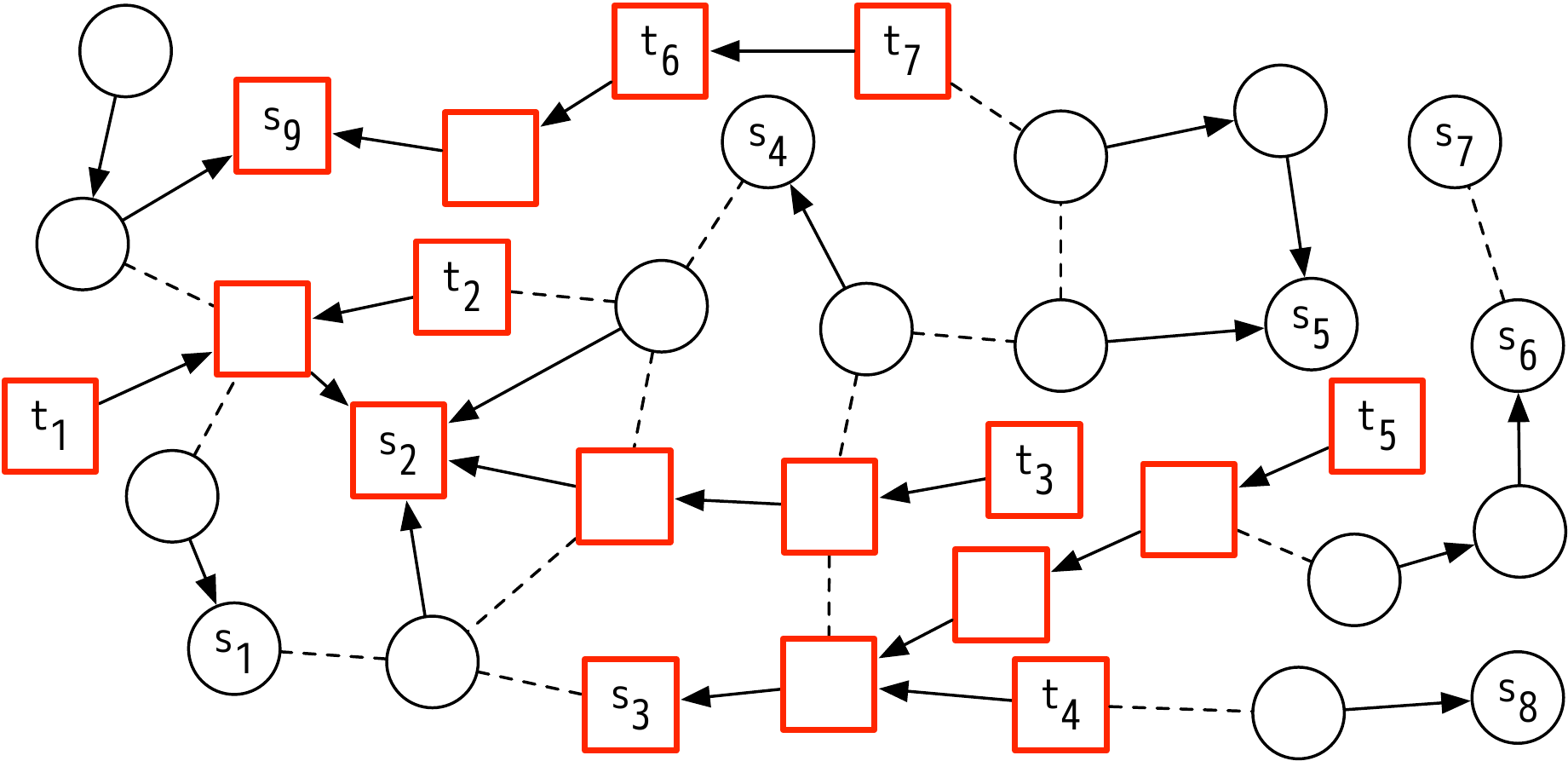}
	\caption{A legitimate configuration of Layer 1.
		Solid arrows represent the generated L1 trees rooted at sender nodes.
		A red square is a red node.
	}
	\label{fig:example-layer1}
\end{figure}
\begin{algorithm}[tp]
\caption{\textsf{L1SpanningForest} for node $v \in V$}
\label{alg:l1algorithm}
\small

\textbf{Input}:\vs

\begin{tabular}{l}
$\calS$: A set of sender nodes \\
$\calT$: A set of target nodes \\
\end{tabular}

\vs
\textbf{Variable}:\vs

\begin{tabular}{l}
$v.\var{l1\_dist} \in \mathbb{N}$: Distance to the root of its L1 tree\\
$v.\var{l1\_parent} \in N(v) \cup \{v\}$: The parent node of its L1 tree\\
$v.\var{l1\_color} \in \{\false, \true\}$: {\true} if $v$ is a red node\\
\end{tabular}

\vs
\textbf{Macro}:\vs

\begin{tabular}{l}
$\macro{NearestL1ParentDist}(v) = \min_{u \in N(v)} u.\var{l1\_dist}$ \\
$\macro{L1CorrectDist}(v) = \left\{ \begin{array}{ll}
	0 & v \in \calS \\
	\macro{NearestL1ParentDist}(v)+1 & \otherwise \\
\end{array} \right.$ \\

$\macro{NearestL1Parent}(v) = \min\left( \argmin_{u \in N(v)} u.\var{l1\_dist} \right)$ \\
$\macro{L1CorrectParent}(v) = \left\{ \begin{array}{ll}
	v & v \in \calS \\
	\macro{NearestL1Parent}(v) & \otherwise \\
\end{array} \right.$ \\

$\macro{RedChild}(v) = $\\
\tab\tab $\{ u \mid u \in N(v) \wedge u.\var{l1\_parent} = v \wedge u.\var{l1\_color} \}$ \\
$\macro{IsRed}(v) = v \in \calT \vee RedChild(v) \neq \emptyset$ \\

\end{tabular}

\vs
\textbf{Action}:\vs

\begin{tabular}{lll}
L1FixDist &::& $v.\var{l1\_dist} \neq \macro{L1CorrectDist}(v)$ \\
&$\ra$& $v.\var{l1\_dist} \gets \macro{L1CorrectDist}(v)$ \\
L1FixParent &::&  $v.\var{l1\_parent} \neq \macro{L1CorrectParent}(v)$ \\
&$\ra$& $v.\var{l1\_parent} \gets \macro{L1CorrectParent}(v)$ \\
L1FixColor &::& $v.\var{l1\_color} \neq \macro{IsRed}(v)$ \\
&$\ra$& $v.\var{l1\_color} \gets \macro{IsRed}(v)$ \\
\end{tabular}

\end{algorithm}
Each sender $s \in \calS$ builds a Breadth-First-Search (BFS) tree on a given graph $G$.
Each non-sender node $v$ joins the tree of its nearest sender.
If there are two or more such trees, $v$ chooses one of the trees to which it can connect through the minimum-label edge.
The construction is done by L1FixDist and L1FixParent actions, and these trees form a BFS spanning forest on graph $G$.
In this layer, each target node $t \in \calT$ changes its color to red by L1FixColor action.
The color change propagates from the target node towards the root (sender) $s \in \calS$ of the tree to tell that there is a target node in $s$'s tree.
L1 trees ensure that every target node $t \in \calT$ is reachable from a sender node by tracing the trees from their root nodes.
Note that an L1 tree whose root node is not red is ignored after this layer.

\textbf{Layer 2 (Fig.~\ref{fig:example-layer2}):}
\begin{figure}[tp]
	\centering
	\includegraphics[scale=0.425]{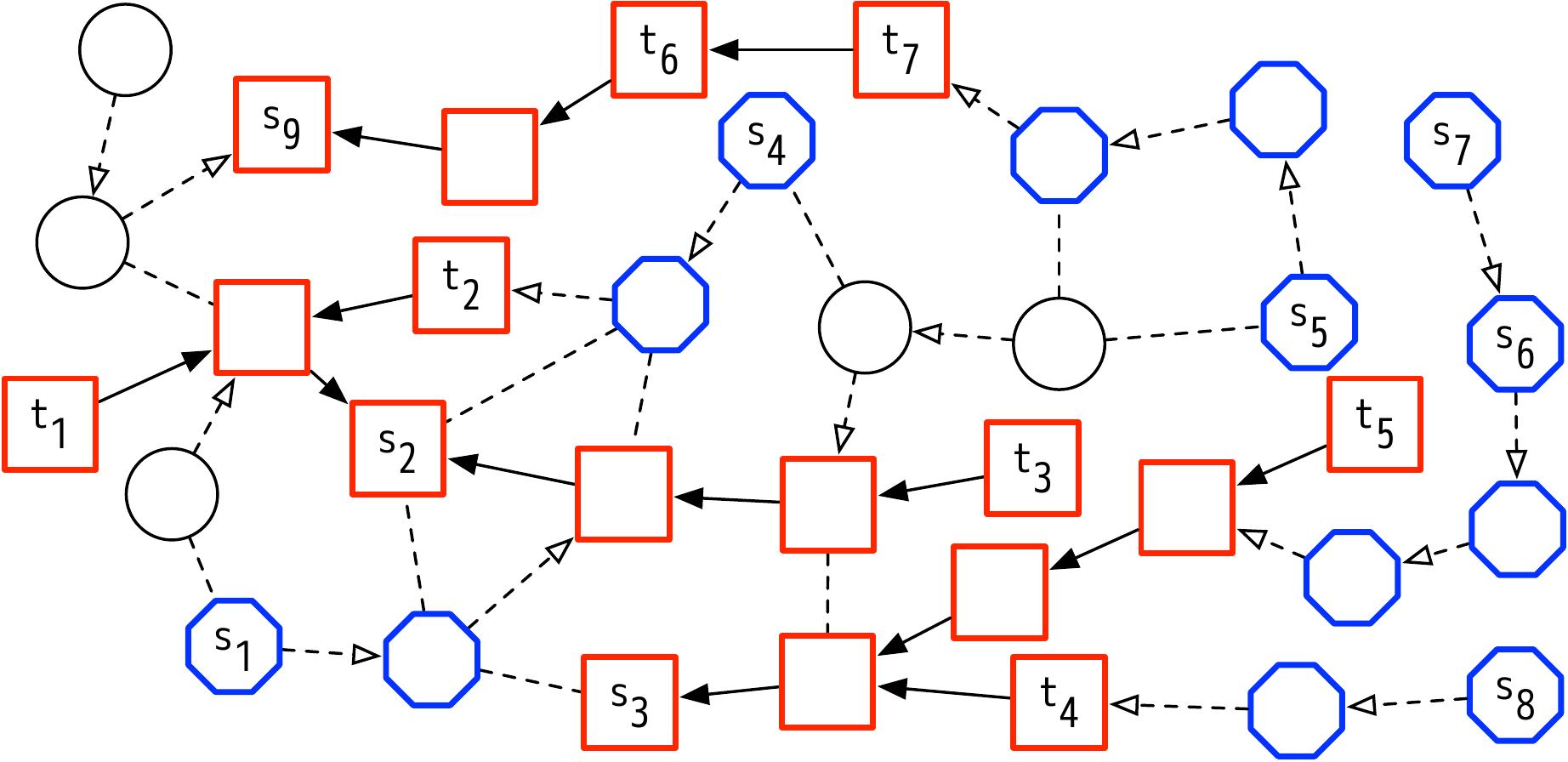}
	\caption{
		A legitimate configuration of layer 2.
		Dashed arrows represent the constructed L2 trees rooted at red nodes.
		A blue octagon is a blue node.
		For simplicity, any arrow of an L1 tree rooted at a colorless sender node is omitted.
	}
	\label{fig:example-layer2}
\end{figure}
\begin{algorithm}[tp]
\caption{\textsf{L2SpanningForest} for node $v \in V$}
\label{alg:l2algorithm}
\small

\textbf{Variable}:\vs

\begin{tabular}{l}
$v.\var{l2\_dist} \in \mathbb{N}$: Distance to the (red) root of L2 tree\\
$v.\var{l2\_parent} \in N(v) \cup \{v\}$: The parent node of its L2 tree\\
$v.\var{l2\_color} \in \{\false, \true\}$: {\true} if $v$ is a blue node\\
\end{tabular}

\vs
\textbf{Macro}:\vs

\begin{tabular}{lll}

$\macro{NearestL2ParentDist}(v) = \min_{u \in N(v)} u.\var{l2\_dist}$ \\
$\macro{L2CorrectDist}(v) = \left\{ \begin{array}{l}
	0 \hfill v.\var{l1\_color} = \red \\
	\macro{NearestL2ParentDist}(v)+1 \tab \otherwise \\
\end{array} \right.$ \\

$\macro{NearestL2Parent}(v) = \min\left( \argmin_{u \in N(v)} u.\var{l2\_dist} \right)$ \\
$\macro{L2CorrectParent}(v) = \left\{ \begin{array}{l}
	v \hfill v.\var{l1\_color} = \red \\
	\macro{NearestL2Parent}(v) \tab \otherwise \\
\end{array} \right.$ \\

$\macro{BlueChild}(v) = $ \\
\tab\tab $\{ u \mid u \in N(v) \wedge u.\var{l2\_parent} = v \wedge u.\var{l2\_color} \}$ \\
$\macro{IsBlue}(v) = \neg v.\var{l1\_color} \wedge (v \in \calS \vee \macro{BlueChild}(v) \neq \emptyset)$ \\

\end{tabular}

\vs
\textbf{Action}:\vs

\begin{tabular}{lll}
L2FixDist &::& $v.\var{l2\_dist} \neq \macro{L2CorrectDist}(v)$ \\
&$\ra$& $v.\var{l2\_dist} \gets \macro{L2CorrectDist}(v)$ \\
L2FixParent &::& $v.\var{l2\_parent} \neq \macro{L2CorrectParent}(v)$ \\
&$\ra$& $v.\var{l2\_parent} \gets \macro{L2CorrectParent}(v)$ \\
L2FixColor &::& $v.\var{l2\_color} \neq \macro{IsBlue}(v)$ \\
&$\ra$& $v.\var{l2\_color} \gets \macro{IsBlue}(v)$ \\
\end{tabular}

\end{algorithm}
Each red node builds another kind of BFS tree on $G$ by L2FixDist and L2FixParent actions to guarantee every color-less sender node, e.g., $s_1$ and $s_7$ in Fig.~\ref{fig:example-layer1}, can reach a red node.
If a node has two or more parent candidates, the node chooses one of them in the same way as for layer 1.
Such a colorless sender node also changes its color to blue by L2FixColor action, and this change propagates from a child to its parent on the L2 trees.
L2 trees ensure that every blue sender node can reach a red node through blue nodes.

\textbf{Layer 3 (Fig.~\ref{fig:example-layer3}):}
\begin{figure}[tp]
	\centering
	\includegraphics[scale=0.425]{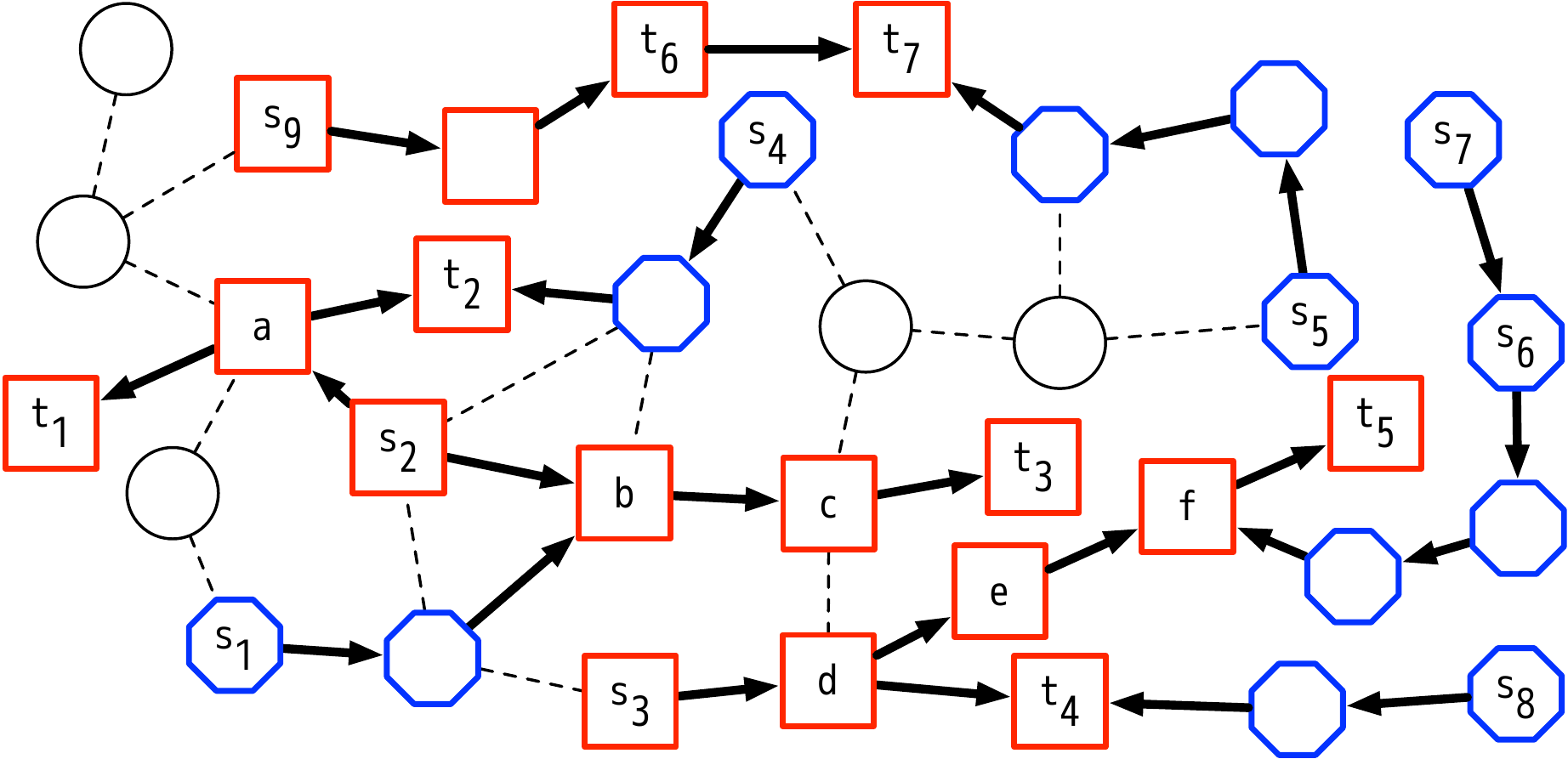}
	\caption{
		A legitimate configuration of layer 3.
		Black bold arrows represent the constructed (possibly non-minimal) weakly $\calS\calT$-reachable DAG.
		For simplicity, any dashed arrow of an L2 tree from a colorless node is omitted.
	}
	\label{fig:example-layer3}
\end{figure}
\begin{algorithm*}[tp]
	\caption{\textsf{L3WSTDAG} for node $v \in V$}
\label{alg:l3algorithm}
\small

\textbf{Variable}:\vs

\begin{tabular}{lll}
$v.\var{l3\_arc}[u] \in \{\false, \true\}$ & : & An arc of (possibly non-minimal) weakly $\calS\calT$-reachable DAG \\
\end{tabular}

\vs
\textbf{Macro}:\vs

\begin{tabular}{lll}
$\macro{HasL3Arc}(v, u)$ & = & $(v.\var{l1\_color} \wedge u \in \macro{RedChild}(v)) \vee (v.\var{l2\_color} \wedge v.\var{l2\_parent} = u)$ \\

$\macro{WrongL3ArcDest}(v)$ & = & $ \{ u \mid u \in N(v) \wedge v.\var{l3\_arc}[u] \neq \macro{HasL3Arc}(v,u) \} $ \\
\end{tabular}

\vs
\textbf{Action}:\vs

\begin{tabular}{lllll}
L3FixArc & :: & $\macro{WrongL3ArcDest}(v) \neq \emptyset$ & $\ra$ & $\forall u \in \macro{WrongL3ArcDest}(v), v.\var{l3\_arc}[u] \gets \macro{HasL3Arc}(v,u)$ \\
\end{tabular}

\end{algorithm*}
A red (resp.~blue) node generates an arc to its red child node on the L1 tree (resp. its parent node of the L2 tree) by L3FixArc action.
The action also removes any wrong arc created by a transient fault.
These arcs construct a (possibly non-minimal) weakly $\calS\calT$-reachable DAG $\vec{G}$.

\textbf{Layer 4 (Fig.~\ref{fig:example-layer4}):}
\begin{figure}[tp]
	\centering
	\includegraphics[scale=0.425]{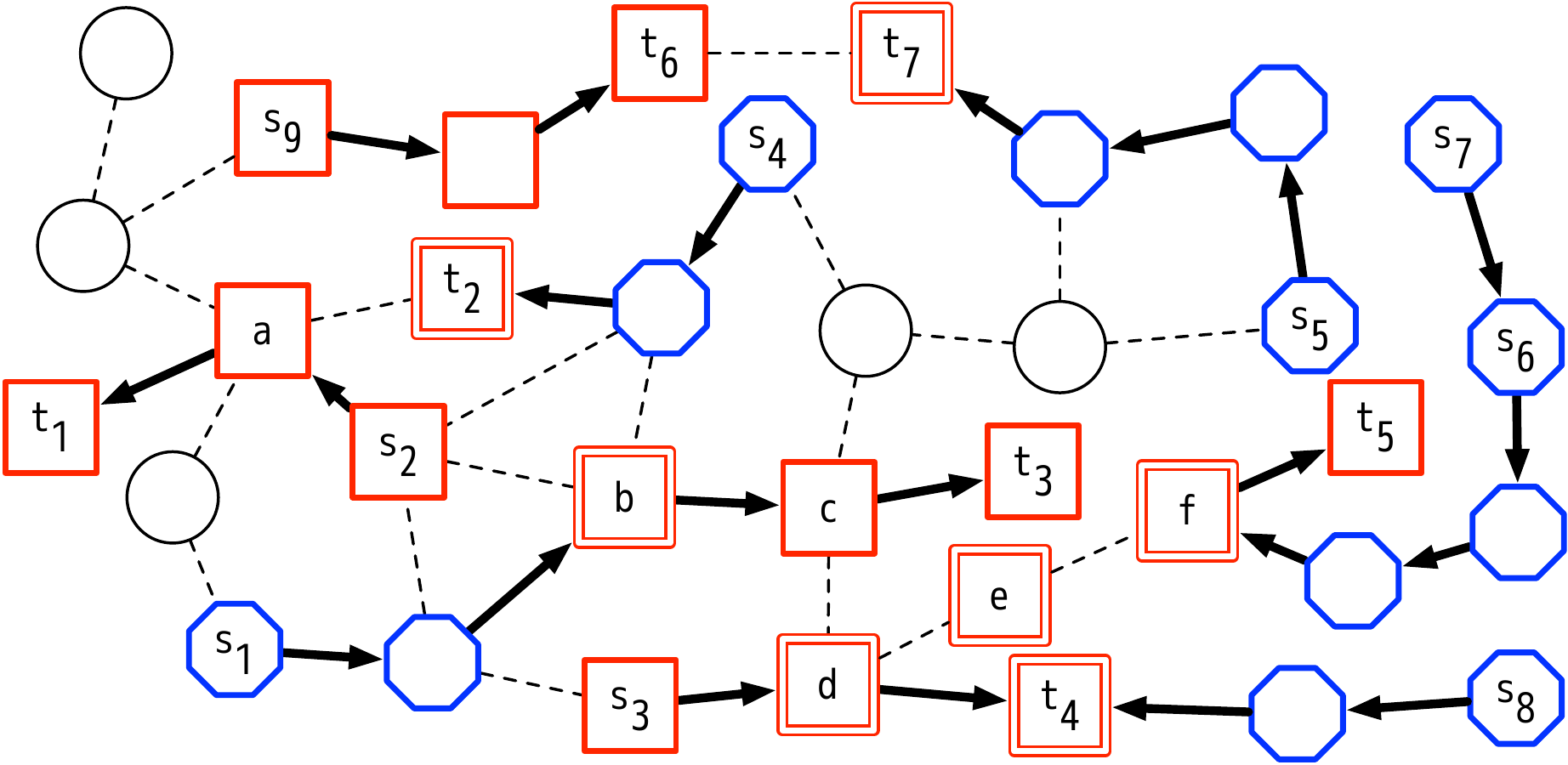}
	\caption{
		A legitimate configuration of layer 4.
		A double red square is a branch node.
	}
	\label{fig:example-layer4}
\end{figure}
\begin{algorithm*}[tp]
\caption{\textsf{L4RedundantArcRemoval} for node $v \in V$}
\label{alg:l4algorithm}
\small

\textbf{Variable}:\vs

\begin{tabular}{lll}
$v.\var{arc}[u] \in \{\false, \true\}$ & : & See Section \ref{sec:problem-spec} for details\\
$v.\var{l4\_branch} \in \{\false, \true\}$ & : & {\true} if $v$ is a branch node\\
\end{tabular}

\vs
\textbf{Macro}:\vs

\begin{tabular}{lll}

$\macro{WrongArcDest}(v)$ & = & $ \{ u \mid u \in N(v) \wedge \neg v.\var{l3\_arc}[u] \wedge v.\var{arc}[u] \} $ \\

$\macro{BranchChild}(v)$ & = & $\{ u \mid u \in N(v) \wedge u.\var{l1\_parent} = v \wedge u.\var{l4\_branch} \}$ \\
$\macro{IsBranch}(v)$ & = & $v.\var{l1\_color} \wedge ( \macro{BlueChild}(v) \neq \emptyset \vee (\macro{BranchChild}(v) = \macro{RedChild}(v) \neq \emptyset) \wedge v \notin \calT )$ \\

$\macro{RemovalRule1}(v,u)$ & = & $u \in \macro{BranchChild}(v) \wedge \macro{BranchChild}(v) \subset \macro{RedChild}(v)$ \\
$\macro{MinBranchChild}(v)$ & = & $\min\left( \argmin_{u \in BranchChild(v)} l_v(u) \right)$ \\
$\macro{RemovalRule2}(v,u)$ & = & $u \in \macro{BranchChild}(v) \wedge \macro{BranchChild}(v) = \macro{RedChild}(v) \wedge u \neq \macro{MinBranchChild}(v)$ \\
$\macro{RemovalRule3}(v,u)$ & = & $u \in \macro{BranchChild}(v) \wedge \macro{BranchChild}(v) = \macro{RedChild}(v) \wedge v \in \calT$ \\
$\macro{RemovalRule4}(v)$ & = & $v.\var{l1\_parent} \neq v \wedge \neg v.\var{l1\_parent}.\var{arc}[v] \wedge BlueChild(v) = \emptyset$ \\
$\macro{IsRedundant}(v,u)$ & = & $\macro{RemoveRule1}(v,u) \vee \macro{RemoveRule2}(v,u) \vee \macro{RemoveRule3}(v,u) \vee \macro{RemoveRule4}(v)$ \\

$\macro{MissingArcDest}(v)$ & = & $ \{ u \mid u \in N(v) \wedge \neg \macro{IsRedundant}(v,u) \wedge \neg v.\var{arc}[u] \wedge v.\var{l3\_arc}[u] \} $ \\
$\macro{RedundantArcDest}(v)$ & = & $ \{ u \mid u \in N(v) \wedge \macro{IsRedundant}(v,u) \wedge v.\var{arc}[u] \} $ \\
\end{tabular}

\vs
\textbf{Action}:\vs

\begin{tabular}{lllll}
L4RemoveWrongArc & :: & $\macro{WrongArcDest}(v) \neq \emptyset$ & $\ra$ & $\forall u \in \macro{WrongArcDest}(v), v.\var{arc}[u] \gets \false$\\
L4FixBranch & :: & $v.\var{l4\_branch} \neq \macro{IsBranch}(v)$ & $\ra$ & $v.\var{l4\_branch} \gets \macro{IsBranch}(v)$ \\
L4AddArc & :: & $\macro{MissingArcDest}(v) \neq \emptyset$ & $\ra$ & $\forall u \in \macro{MissingArcDest}(v), v.\var{arc}[u] \gets \true$ \\
L4RemoveRedundantArc & :: & $\macro{RedundantArcDest}(v) \neq \emptyset $ & $\ra$ & $\forall u \in \macro{RedundantArcDest}(v), v.\var{arc}[u] \gets \false$ \\
\end{tabular}

\end{algorithm*}
In this layer, a node first removes every wrong arc that does not exist in the weakly $\calS\calT$-reachable DAG constructed in Layer 3 by L4RemoveWrongArc.
After that, a red node checks whether it is a \emph{branch} by L4FixBranch action.
A red node is called branch if (i) the node has an incoming arc from a blue node, or (ii) all children in its L1 tree are branch nodes.
For example, in Fig.~\ref{fig:example-layer3}, $t_2$ and $d$ become branch nodes for condition (i) and (ii), respectively.
Redundant arcs are removed based on the existence of branch nodes in neighbor nodes.
Then, a node regenerates an arc that is missing in the output network by L4AddArc.
Finally, a red node $v$ removes every redundant outgoing arc by L4RemoveRedundantArc to guarantee the minimality requirement of the weakly $\calS\calT$-reachable DAG.
There are four removal rules:
\begin{itemize}
	\item \textbf{Rule 1:} if a part of $v$'s child nodes are branch nodes, $v$ removes all arcs to the branch nodes (nodes $s_2$ and $a$ in Fig.~\ref{fig:example-layer4}).
	\item \textbf{Rule 2:} if all of $v$'s child nodes of its L1 tree are branch nodes, $v$ removes all arcs to the nodes except for the arc $(v,u)$ where $u$ has the minimum label in $N(v)$ (node $d$ in Fig.~\ref{fig:example-layer4}).
	\item \textbf{Rule 3:} if all of $v$'s child nodes of its L1 tree are branch nodes and $v$ is a target node, $v$ removes all arcs (node $t_6$ in Fig.~\ref{fig:example-layer4}).
	\item \textbf{Rule 4:} if $v$ has neither an incoming arc from its parent nor an incoming arc from a blue node, $v$ removes all arcs to its child nodes of the L1 tree (nodes $e$ and $f$ after $d$ removes $(d,e)$ in Fig.~\ref{fig:example-layer4}).
\end{itemize}
After removing all redundant arcs, the remaining arcs form a minimal weakly $\calS\calT$-reachable DAG $\vec{G}$, as depicted in Fig.~\ref{fig:example-finish}.
\begin{figure}[tp]
	\centering
	\includegraphics[scale=0.425]{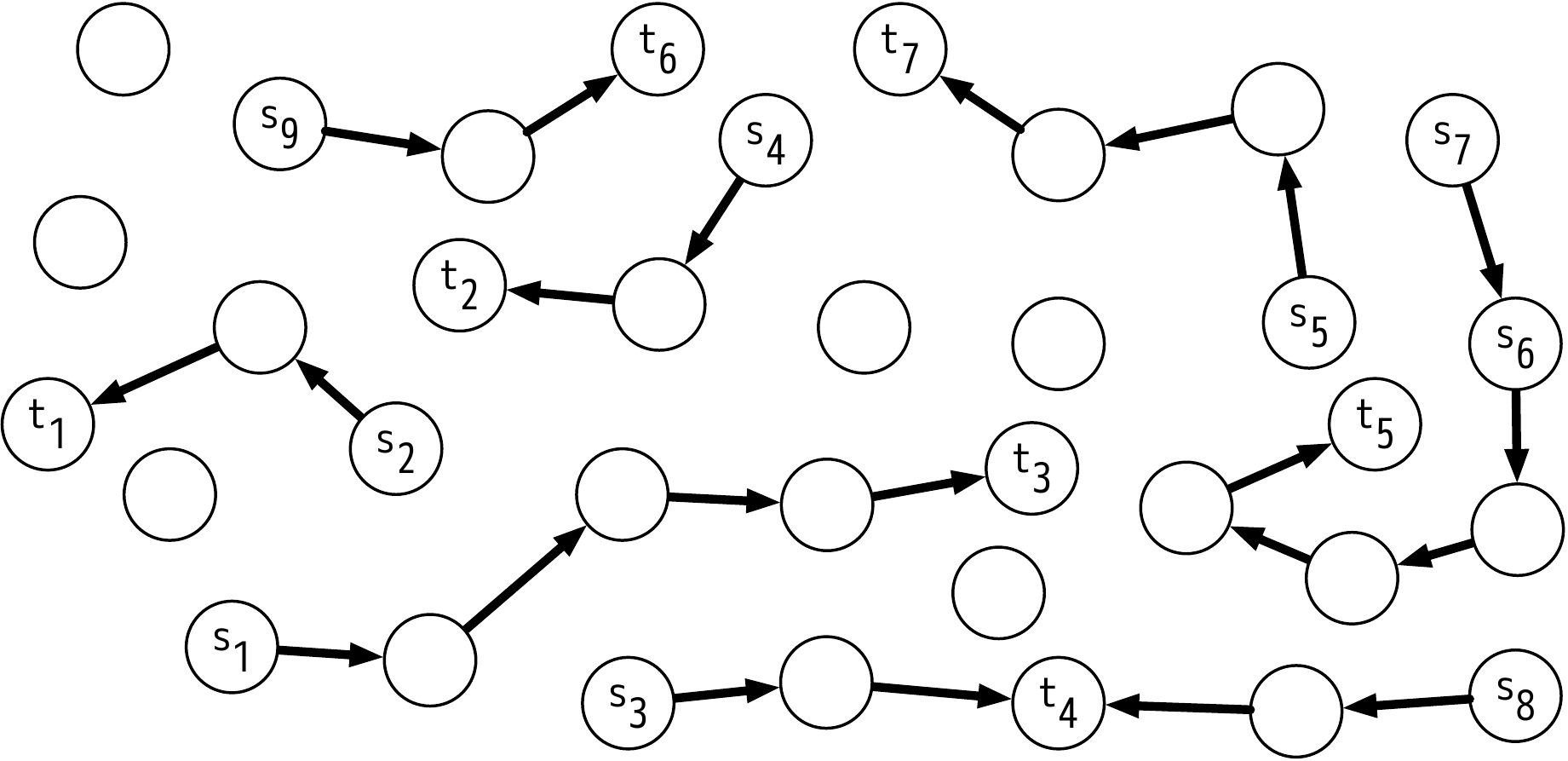}
	\caption{The constructed minimal weakly $\calS\calT$-reachable DAG}
	\label{fig:example-finish}
\end{figure}

\section{Correctness}
\label{sec:correctness}

Here, we show the correctness of Algorithm \textsf{MWSTDAG} by proving the correctness of each layer.

\begin{lemma}
\label{lem:l1-convergence}
From any initial configuration, Algorithm \textsf{MWSTDAG} eventually reaches a legitimate configuration for layer 1, in which,
(1) every node belongs to the nearest BFS tree rooted at a sender node, and 
(2) a node has a red color if and only if the node can reach at least one target node on its L1 tree.
\end{lemma}
Similarly, we can also prove the correctness of layer 2 algorithm.
\begin{lemma}
\label{lem:l2-convergence}
From any legitimate configuration for layer 1,
Algorithm \textsf{MWSTDAG} eventually reaches a legitimate configuration for layer 2, in which
(1) every non-red node belongs to the nearest tree rooted at a red node, and 
(2) every node that can reach a red node on the tree has a blue color.
\end{lemma}

Based on Lemmas \ref{lem:l1-convergence} and \ref{lem:l2-convergence}, we show that the layer 3 algorithm constructs a weakly $\calS\calT$-reachable DAG.
\begin{lemma}
\label{lem:l3-convergence}
From any legitimate configuration for layer 2, Algorithm \textsf{MWSTDAG} eventually reaches a legitimate configuration for layer 3, in which $\var{l3\_arc}$ of each node forms a weakly $\calS\calT$-reachable DAG $\vec{G}$ of graph $G$.
\end{lemma}

Previously, we proved the correctness of layers 1--3.
Before proving the correctness of layer 4, we prove the following two supplementary lemmas.

\begin{lemma}
\label{lem:l4-branch}
From any legitimate configuration for layer 3,
Algorithm \textsf{MWSTDAG} reaches a configuration in which (1) every red node that has a blue child node on its L2 tree is a branch node, (2) every non-target red node whose all the child nodes in its L1 tree are red branch nodes is a branch node, and (3) every other red node is not a branch node.
\end{lemma}

\begin{lemma}
\label{lem:removal-safety}
In a legitimate configuration for layer 3 where only correct nodes are branch nodes, any removal of an arc by L4RemoveRedundantArc action does not break the conditions C1 and C2 of the weakly $\calS\calT$-reachable DAG.
\end{lemma}

\begin{proof}
In L4RemoveRedundantArc action, there are four rules, as depicted in Fig.~\ref{fig:removal-rules}.
Hereafter, we will verify that each rule removes only a redundant arc, and the conditions C1 and C2 still hold after the removal.

\begin{figure}[tp]
	\centering
	\begin{tabular}{ccc}
		\includegraphics[scale=0.425]{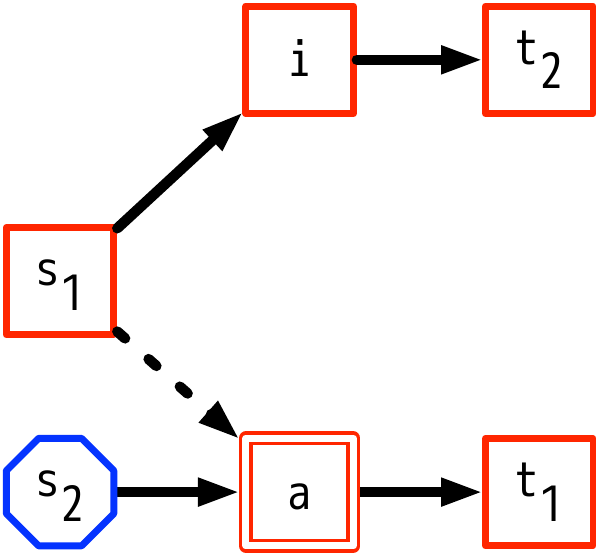} & & \includegraphics[scale=0.425]{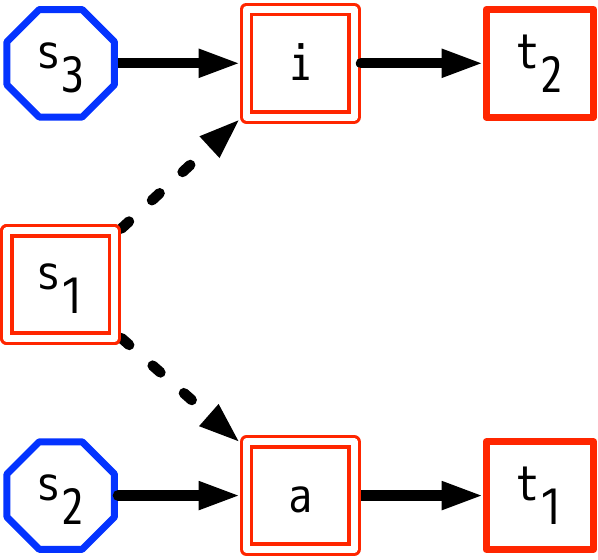} \\
		(a) Rule 1 & & (b) Rule 2 \\
		& & \\
		\includegraphics[scale=0.425]{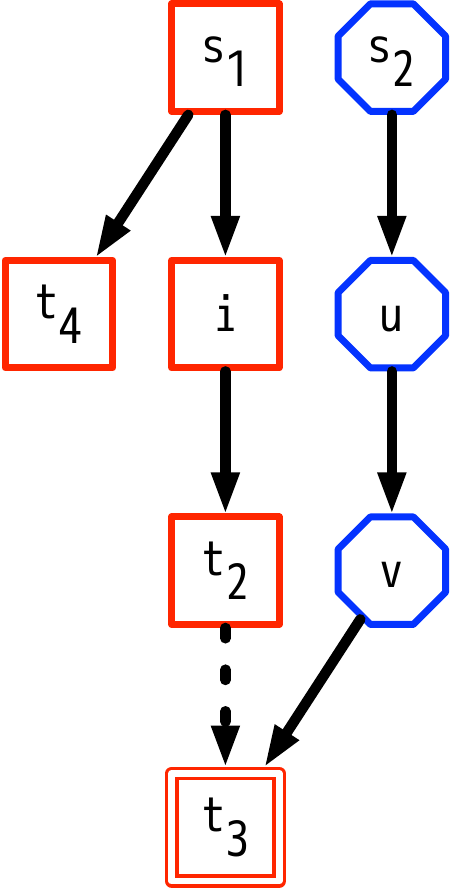} & & \includegraphics[scale=0.425]{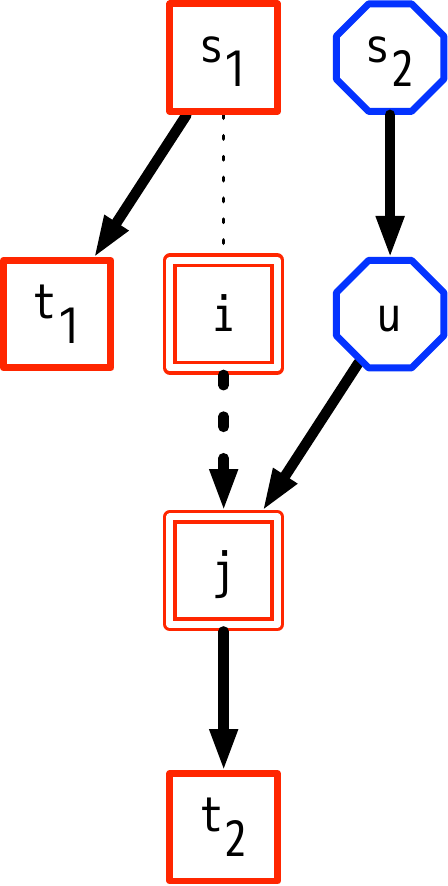} \\
		(c) Rule 3 & & (d) Rule 4 \\
	\end{tabular}
	\caption{
		Removal rules of Algorithm \textsf{MWSTDAG} for a redundant arc.
		A dashed arrow will be removed by a rule.
	}
	\label{fig:removal-rules}
\end{figure}

By Rule 1, a node removes every outgoing arc to a branch node if the node has an arc to a non-branch node.
In Fig.~\ref{fig:removal-rules}(a), arc $(s_1,a)$ is removed by $s_1$.
However, $s_1$ and $s_2$ can reach $t_2$ and $t_1$, respectively.
Therefore, the conditions still hold.

If all child nodes are branch nodes, by Rule 2, a node removes every outgoing arc to branch nodes except for a child branch node with a minimum label among the branch nodes.
In Fig.~\ref{fig:removal-rules}(b), $(s_1,a)$ or $(s_1,i)$ is removed based on their labels, while the conditions still hold because $s_1$ can reach either nodes $t_1$ or $t_2$, and $t_1$ and $t_2$ are still reachable from $s_2$ and $s_3$, respectively.

Rule 3 removes every outgoing arc from a target node to a branch node if all child nodes are branch nodes.
In Fig.~\ref{fig:removal-rules}(c), $(t_2,t_3)$ is removed by $t_2$ because this arc is redundant.
However, $t_3$ is still reachable from $s_2$, and this removal does not break the conditions.

Rule 4 is different from the other rules and cleans up an outgoing arc from a node having no incoming arc.
Such a situation happens when an arc is removed by Rules 1, 2, and 3.
In Fig.~\ref{fig:removal-rules}(d), node $i$ has no incoming arc, so the node removes $(i,j)$.
Since there is no incoming arc to $i$, the conditions hold.

We confirmed that, after applying every rule, the conditions C1 and C2 still held, and therefore, this lemma holds.
\end{proof}

Next, we prove the correctness of layer 4 with Lemmas \ref{lem:l4-branch} and \ref{lem:removal-safety}.

\begin{lemma}
\label{lem:l4-convergence}
From any legitimate configuration for layer 3, Algorithm \textsf{MWSTDAG} eventually reaches a legitimate configuration for layer 4, in which a minimal weakly $\calS\calT$-reachable DAG is constructed.
\end{lemma}
\begin{proof}
\begin{figure}[tp]
	\centering
	\includegraphics[scale=0.425]{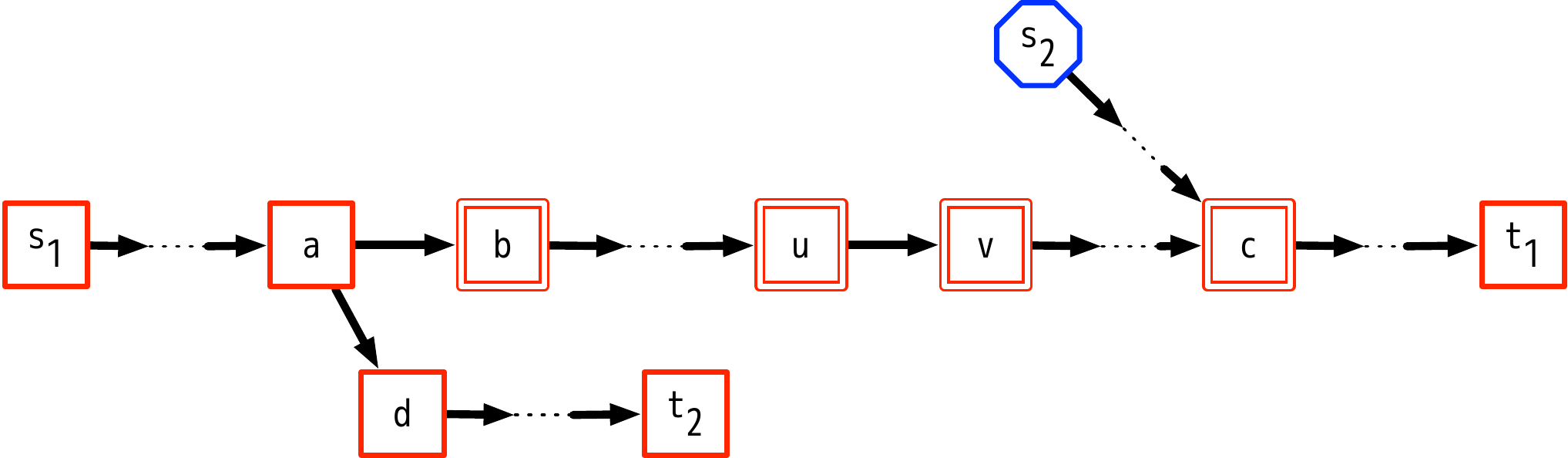}\\
	(a) Case 1\\
	\includegraphics[scale=0.425]{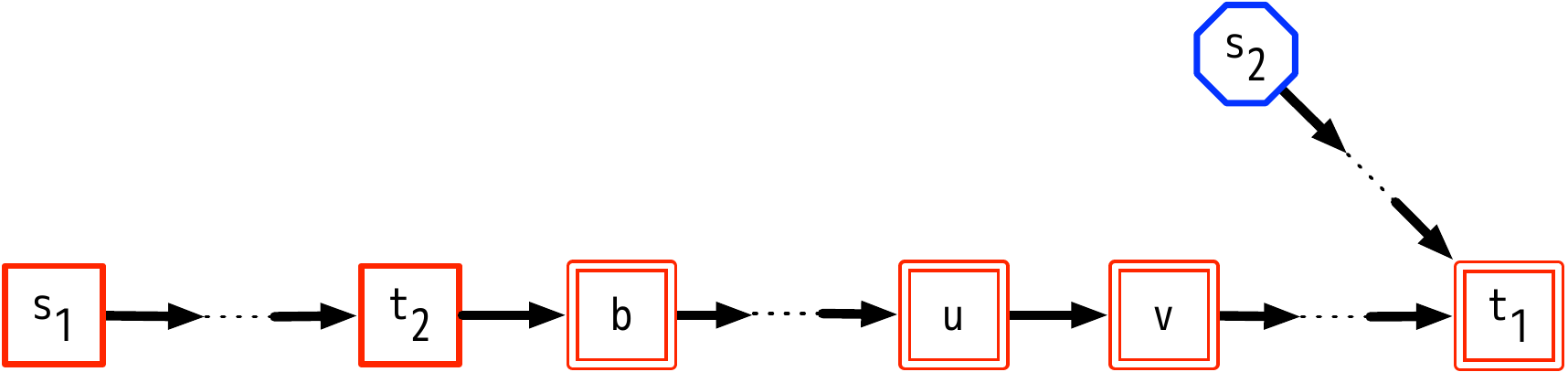}\\
	(b) Case 2\\
	\caption{Proof assumption of Lemma \ref{lem:l4-convergence}}
	\label{fig:minimality-proof}
\end{figure}

Lemma \ref{lem:l3-convergence} proved that Algorithm \textsf{MWSTDAG} eventually constructs a (possibly non-minimal) weakly $\calS\calT$-reachable DAG $\vec{G}$ from any layer 2 legitimate configuration, and Lemma \ref{lem:l4-branch} ensures that every red node satisfying one of the conditions of a branch node eventually sets \var{l4\_branch} to \true.
In addition, Lemma \ref{lem:removal-safety} proves that, after removing any redundant arcs of $\vec{G}$ by L4RemoveRedundantArc action, the resulting DAG $\vec{G'}$ is still a weakly $\calS\calT$-reachable DAG.
Therefore, the remaining question is whether $\vec{G'}$ is minimal or not.
To answer the question positively, we assume that there is an arc $(u, v)$ that can be removed without breaking conditions C1 and C2 and prove by contradiction that such an arc does not exist.

Since $(u, v)$ is redundant and can be removed safely, there must be a sender node $s_1$ and a target node $t_1$ that have a path $(s_1, \dots, u, v, \dots, t_1)$ through arc $(u,v)$ on the constructed minimal weakly $\calS\calT$-reachable DAG.
In addition, $s_1$ can reach another target node, $t_2$, and $t_1$ is reachable from another sender node $s_2$.
In this case, the intermediate node $c$ in the path that has an incoming arc from $s_2$ must be a branch node because $s_2$ is a blue node, and its parent and ancestor nodes also become branch nodes by L4FixBranch action.
Thus, at least nodes $b, \dots, u, v, \dots, c$ are branch nodes, where $b$ is the node between $s_1$ and $u$ in the path.
We can consider two cases for this situation, as illustrated in Fig.~\ref{fig:minimality-proof}.
Note that case 2 of Fig.~\ref{fig:minimality-proof} is the special case of case 1 when $a=t_2$ and $c=t_1$.

For case 1 (Fig.~\ref{fig:minimality-proof}(a)), the parent node $a$ of $b$ applies Rules 1 or 2 depending on whether $d$ is a branch or not, as follows:
\begin{itemize}
\item \textbf{(i) $d$ is a branch node:}
Since all child nodes of $a$ are branch nodes, node $a$ removes one of two arcs $(a, b)$ or $(a, d)$ based on the labels of $b$ and $d$ on $a$ by Rule 2.
However, node $a$ cannot remove $(a,d)$ because this removal makes $(u,v)$ not redundant, which contradicts the assumption.
Therefore, the removed arc must be $(a,b)$.
In this case, node $b$ realizes that it has no incoming arc and removes its outgoing arc by Rule 4.
This removal propagates from node $b$ to the parent of node $c$, including node $u$.
As a result, arc $(u,v)$ is removed, which contradicts the assumption.

\item \textbf{(ii) $d$ is not a branch node:}
By Rule 1, node $a$ removed all outgoing arcs to child branch nodes, including $(a,b)$.
After this removal, as with case (i), $(u,v)$ is also removed.
This leads to a contradiction.
\end{itemize}

For case 2 (Fig.~\ref{fig:minimality-proof}(b)), target node $t_2$ removes $(t_2,b)$ by Rule 3, and then, the remaining arcs from $b$ to $t_1$, including $(u,v)$, are also removed by Rule 4.
This is a contradiction.

Both the cases contradict the assumption.
Therefore, we can conclude that the constructed weakly $\calS\calT$-reachable DAG $\vec{G'}$ is minimal.
\end{proof}

The following lemma proves the preferable property of the proposed algorithm after reaching legitimate configurations.
\begin{lemma}
\label{lem:silentness}
Algorithm \textsf{MWSTDAG} is silent.
\end{lemma}

We also prove the time and space complexities of Algorithm \textsf{MWSTDAG} with the following two lemmas.
\begin{lemma}
\label{lem:timecomplexity}
From any initial configuration on a given simple and connected graph $G$ whose diameter is $D$, Algorithm \textsf{MWSTDAG} constructs a minimal weakly $\calS\calT$-reachable DAG within $O(D)$ rounds.
\end{lemma}
\begin{proof}
In layer 1, each sender node $s \in \calS$ builds a BFS tree, and this tree construction requires at most $D$ rounds for propagating correct \var{l1\_dist} values, one round for fixing \var{l1\_parent}, and at most $D$ rounds for fixing \var{l1\_color}; thus, the layer 1 algorithm requires $O(D)$ rounds in total.
Similarly, the layer 2 algorithm also requires $O(D)$ rounds.
In layer 3, nodes construct a weakly $\calS\calT$-reachable DAG in one round after stabilizing layer 2.
Finally, layer 4 removes any wrong arcs in one round, fixes \var{l4\_branch} in $O(D)$ rounds, adds missing arcs in one round, and, finally, removes redundant arcs in $O(D)$ rounds; so, the layer 4 algorithm requires $O(D)$ rounds in total.
Therefore, Algorithm \textsf{MWSTDAG} can construct a minimal weakly $\calS\calT$-reachable DAG within $O(D)$ rounds.
\end{proof}

\begin{lemma}
\label{lem:spacecomplexity}
Each node requires $O(\log D + \Delta)$ bits memory for Algorithm \textsf{MWSTDAG} where $D$ is the diameter of a given graph, and $\Delta$ is the maximum degree of the graph.
\end{lemma}
\begin{proof}
In the layer 1 algorithm, each node has three variables \var{l1\_dist}, \var{l1\_parent}, and \var{l1\_color}, and these variables need $\log D + \log \Delta + 1$ bits in total.
Similarly, the layer 2 algorithm requires $\log D + \log \Delta + 1$ bits.
The layer 3 and layer 4 algorithms have bit arrays of size $\Delta$, \var{l3\_arc} and \var{arc}, respectively.
The layer 4 algorithm also has a one-bit flag \var{l4\_branch}.
Therefore, the space complexity of Algorithm \textsf{MWSTDAG} is $O(\log D + \Delta)$.
\end{proof}

From Lemmas \ref{lem:l1-convergence}, \ref{lem:l2-convergence}, \ref{lem:l3-convergence}, \ref{lem:l4-convergence}, \ref{lem:silentness}, \ref{lem:timecomplexity}, and \ref{lem:spacecomplexity}, we have the following theorem finally.
\begin{theorem}
\label{thm:main-result}
Algorithm \textsf{MWSTDAG} is a silent self-stabilizing algorithm for the minimal weakly $\calS\calT$-reachable DAG construction problem.
Starting from any configuration, every fair execution of the algorithm reaches a final configuration within $O(D)$ rounds.
The algorithm requires $O(\log D + \Delta)$ bit memory for each node.
\end{theorem}

\section{Evaluation}
\label{sec:evaluation}
We conduct small simulations to evaluate the performance of Algorithm \textsf{MWSTDAG}.
In this simulation, we consider a $d \times d$ grid network because we can easily control its diameter $D$.
Note that $n=d^2$, $m=2d(d-1)$, and $D=2d-2$.
We conduct the simulation with parameter $d=6,8,\dots,86$, resulting in $D=10,14,\dots,170$.
We also change the numbers of sender and target nodes, $|\calS|$ and $|\calT|$, from 5 to 15 to observe how these changes affect the total number of rounds required to reach a legitimate configuration of the algorithm.
At the beginning of each iteration of the simulation, we choose the sender and target nodes uniformly at random from the $n$ nodes without overlapping sender nodes and target nodes.
For simplicity, we assume a synchronous execution, in which every enabled node executes their action in every step.
The state of each node is randomly initialized to one of all the possible states at the beginning of each iteration to imitate a transient failure.
We run 500 iterations for each parameter setting and show the average as its result.

\begin{figure}[tp]
	\centering
	\includegraphics[width=70mm]{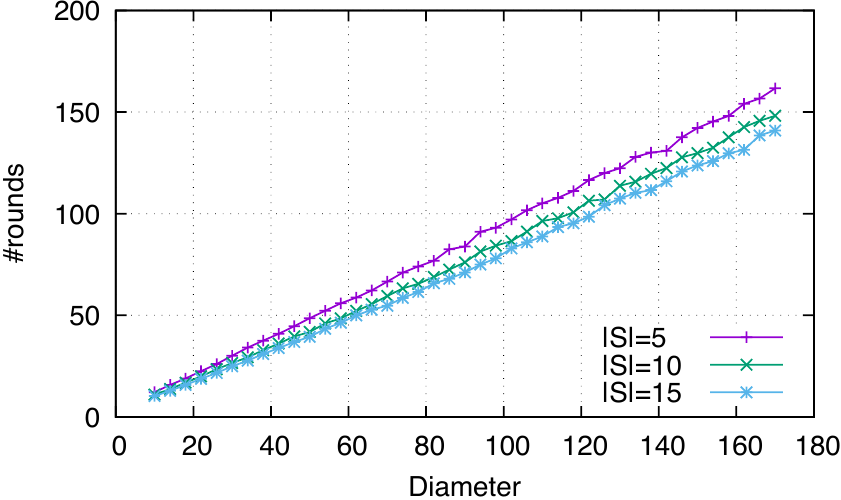}
	\caption{The total number of rounds when the size of sender nodes $|\calS|$ varied from 5 to 15.
		The size of target nodes $|\calT|$ was fixed at 10.}
	\label{fig:n-s}
\end{figure}
\begin{figure}[tp]
	\centering
	\includegraphics[width=70mm]{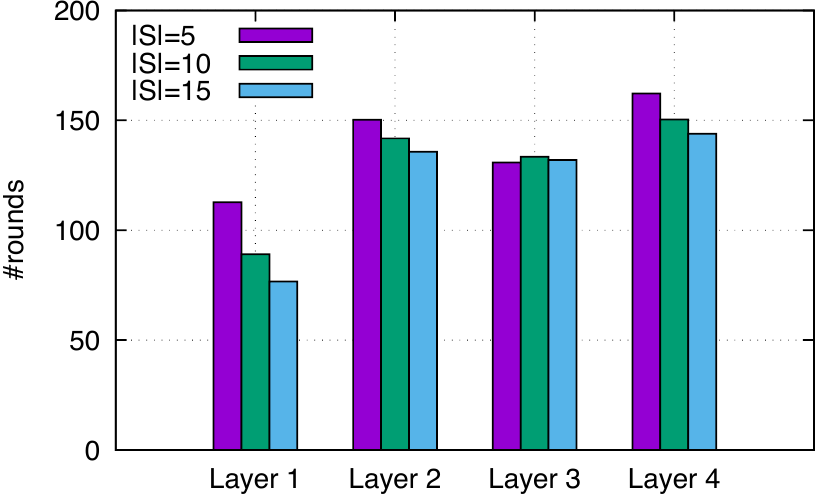}
	\caption{Running times of layers 1--4 when the size of the sender nodes $|\calS|$ varied from 5 to 15.
		The diameter and the size of target nodes $|\calT|$ were fixed at 170 and 10, respectively.}
	\label{fig:steps-varying-s}
\end{figure}

Figure \ref{fig:n-s} shows the number of rounds needed to construct a minimal weakly $\calS\calT$-reachable DAG on grid networks with $|\calS|=5,10,15$ and with a fixed $|\calT|=10$.
The results show that the total number of rounds increased as their diameter increased.
We also observed that the algorithm reached a legitimate configuration with small rounds when the size of the sender nodes was large.

To see where the difference came from, we investigated the running times of each layer.
Here, we define the running time of layer $l$ $(1 \leq l \leq 4)$ as the total number of rounds where at least one node executed layer $l$ actions\footnote{Note that actions of two or more layers may be executed in a round; thus, the sum of all running times is not equal to the total number of rounds in Fig.~\ref{fig:n-s}.}.
Figure \ref{fig:steps-varying-s} shows the running times where the diameter of a grid network was 170 (i.e., $n = (170/2 + 1)^2 = 7396$).
As we can see, the running times of each layer decreased as the size of sender nodes $|\calS|$ increased except for layer 3.
The largest difference happened in layer 1.
This layer finished 76.72 rounds on average when $|\calS| = 15$, whereas it required 112.7 rounds on average when $|\calS| = 5$.
This is because distances between target nodes and sender nodes (i.e., depths of L1 trees) get smaller when there are many sender nodes.

\begin{figure}[tp]
	\centering
	\includegraphics[width=70mm]{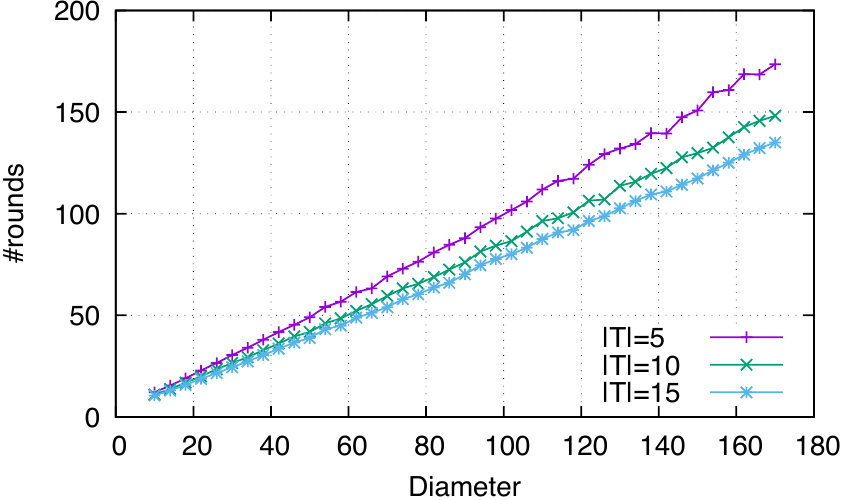}
	\caption{The total number of rounds when the size of target nodes $|\calT|$ varied from 5 to 15.
		The size of sender nodes $|\calS|$ was fixed at 10.}
	\label{fig:n-t}
\end{figure}
\begin{figure}[tp]
	\centering
	\includegraphics[width=70mm]{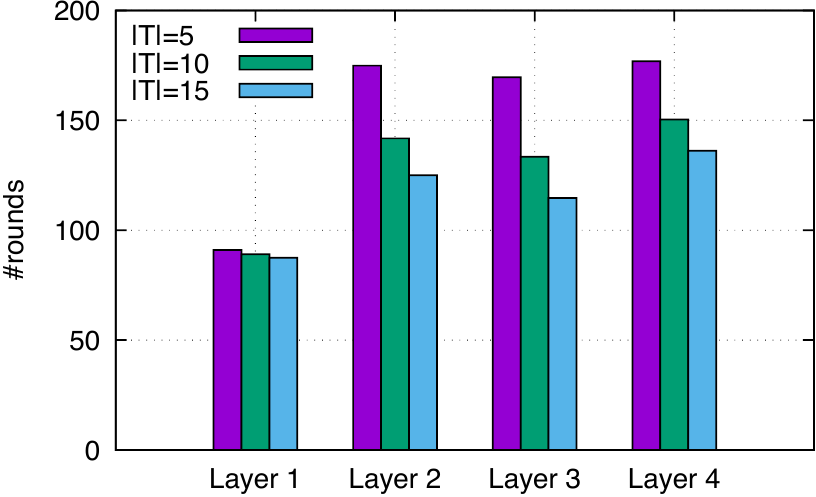}
	\caption{Running times of layers 1--4 when the size of target nodes $|\calT|$ varied from 5 to 15.
		The diameter and the size of sender nodes $|\calS|$ were fixed at 170 and 10, respectively.}
	\label{fig:steps-varying-t}
\end{figure}
Figures \ref{fig:n-t} and \ref{fig:steps-varying-t} show the results of the simulations with $|\calT|=5,10,15$ and with a fixed $|\calS|=10$.
Figure \ref{fig:n-t} shows similar trends to Fig.~\ref{fig:n-s}, but the trends of Fig.~\ref{fig:steps-varying-t} are quite different from that of Fig.~\ref{fig:steps-varying-s}.
The running times of layers 2--4 changed markedly with a different size of $\calT$, while that of layer 1 is almost independent of the size of $\calT$.
Indeed, the running times of layer 2 were 125.0 and 174.9 rounds when $|\calT| = 15$ and $|\calT| = 5$, respectively.
The reason is as follows.
If there are fewer target nodes than sender nodes, most sender nodes cannot reach target nodes and do not become red nodes in layer 1.
Therefore, the sender nodes and intermediate nodes between the senders and red nodes must execute layer 2 actions many times to construct L2 trees and to become blue nodes.
In contrast, red nodes can finish the layer 2 algorithm immediately.
This gap remained in layers 3 and 4.

Figure \ref{fig:enabled-node-transition} shows the transition of the number of enabled nodes in a representative execution where $D=170$ $(d=86)$ and $|\calS|=|\calT|=10$.
The number of the nodes enabled by layer 1 actions decreased rapidly, and the layer 1 algorithm reached its legitimate configuration at round 88.
The decrease speed of the layer 2 algorithm was slower than that of the layer 1.
The layer 2 algorithm took many rounds until reaching its legitimate configuration and terminated at round 124.
Surprisingly, the layer 3 algorithm terminated at round 106 before terminating the layer 2 algorithm.
This is not the special case of the execution, and we frequently observed this situation in other executions.
Actually, the average termination round of the layer 3 algorithm was 117.3, while that of the layer 2 algorithm was 141.7 among the 500 executions.
This was caused by nodes that do not have any arcs in layers 3 and 4.
These nodes are typically far from any sender node and any target node; thus, the correct values of their variables are propagated slowly.
Therefore, the nodes executed layer 2 actions many times based on their wrong values, and the termination of the layer 2 algorithm was late, while the termination of the layer 3 algorithm is not affected by these nodes.
\begin{figure}[tp]
	\centering
	\includegraphics[width=70mm]{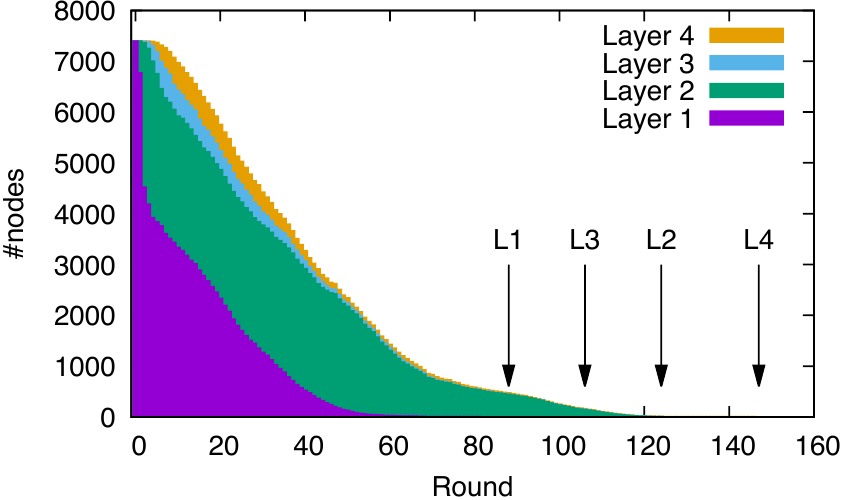}
	\caption{The transition of the number of enabled nodes.
		The arrows in the figure indicate the rounds when each layer algorithm reaches its legitimate configuration.}
	\label{fig:enabled-node-transition}
\end{figure}

\section{Conclusion}
\label{sec:conclusion}

We proposed a self-stabilizing algorithm named \textsf{MWSTDAG} that constructs a minimal weakly $\calS\calT$-reachable directed acyclic graph on a given connected undirected graph and the sets of sender nodes and target nodes, $\calS$ and $\calT$.
This graph guarantees that every sender node $s \in \calS$ can reach at least one target node in $\calT$, every target node $t \in \calT$ is reachable from at least one sender node in $\calS$, and the graph has no directed cycles while keeping the number of arcs in the graph minimal.
To the best of our knowledge, this is the first algorithm that can construct such kinds of DAGs without the restriction of the numbers of sender and target nodes.
The algorithm takes $O(D)$ asynchronous rounds and requires $O(\log D + \Delta)$ bits memory per node for the construction, where $D$ and $\Delta$ are the diameter and the maximum degree of a given graph, respectively.
We also conducted small simulations to evaluate the performance of the proposed algorithm.
The simulation results showed that the total number of rounds increases as the diameter of a network increases, and the execution time decreases when the number of sender nodes or target nodes is large.

For future work, we plan to prove the correctness of the proposed algorithm under a distributed unfair daemon.

\bibliographystyle{abbrv}
\bibliography{collection}

\end{document}